\documentclass[11pt]{article}
\usepackage{mathrsfs}
\usepackage{color}
\usepackage{enumerate}
\usepackage{mathrsfs}
\usepackage{amssymb}
\usepackage{amsfonts}
\usepackage{amsmath}
\usepackage{multirow}
\usepackage{amsthm}
\usepackage{url}
\usepackage{tikz}

\newtheorem{theorem}{Theorem}[section]
\newtheorem{definition}{Definition}[section]
\newtheorem{lemma}[theorem]{Lemma}

\newtheorem{proposition}[theorem]{Proposition}

\newtheorem{remark}{Remark}[section]

\def\f#1{{\mathbb{F}}_{#1}}
\title{New Constructions of MDS Symbol-Pair Codes}

\author{Baokun Ding$^{\text{a}}$, Gennian Ge$^{\text{b,c,}}$\thanks{Corresponding author (e-mail: gnge@zju.edu.cn). Research supported by the National Natural Science Foundation of China under Grant Nos. 11431003 and 61571310.}, Jun Zhang$^{\text{b}}$, Tao Zhang$^{\text{a}}$ and Yiwei Zhang$^{\text{b}}$\\
\footnotesize $^{\text{a}}$ School of Mathematical Sciences, Zhejiang University, Hangzhou 310027, China\\
\footnotesize $^{\text{b}}$ School of Mathematical Sciences, Capital Normal University, Beijing 100048, China\\
\footnotesize $^{\text{c}}$ Beijing Center for Mathematics and Information Interdisciplinary Sciences, Beijing 100048, China.\\}

\begin{document}

\date{}\maketitle
\begin{abstract}
Motivated by the application of high-density data storage technologies, symbol-pair codes are proposed to protect against pair-errors in symbol-pair channels, whose outputs are overlapping pairs of symbols. The research of symbol-pair codes with the largest minimum pair-distance is interesting since such codes have the best possible error-correcting capability. A symbol-pair code attaining the maximal minimum pair-distance is called a maximum distance separable (MDS) symbol-pair code. In this paper, we focus on constructing linear MDS symbol-pair codes over the finite field $\mathbb{F}_{q}$. We show that a linear  MDS symbol-pair code over $\mathbb{F}_{q}$  with pair-distance $5$ exists if and only if the length $n$ ranges from $5$ to $q^2+q+1$. As for codes with  pair-distance $6$, length ranging from $6$ to $q^{2}+1$, we construct linear MDS symbol-pair codes by using a configuration called ovoid in projective geometry. With the help of elliptic curves, we present a construction of linear MDS symbol-pair codes for any pair-distance $d+2$ with length $n$ satisfying $7\le d+2\leq n\le q+\lfloor 2\sqrt{q}\rfloor+\delta(q)-3$, where $\delta(q)=0$ or $1$.
\end{abstract}
\medskip
\noindent {{\it Key words\/}: Symbol-pair read channels, MDS symbol-pair codes, projective geometry, elliptic curves.}
\noindent {\it Mathematics subject classifications\/}: 94B25, 94B60. 
\smallskip
\section{Introduction}
With the development of high-density data storage technologies, while the codes are defined as usual over some discrete symbol alphabet, their reading from the channel is performed as overlapping pairs of symbols. A channel whose outputs are overlapping pairs of symbols is called a symbol-pair channel. A pair-error is defined as a pair-read in which one or more of the symbols are read in error. The design of codes to protect efficiently against a certain number of pair-errors is significant.

Cassuto and Blaum first studied codes that protect against pair-errors in \cite{CB}, as well as pair-error correctability conditions, code construction and decoding, and lower and upper bounds on code sizes.  Later, Cassuto and Litsyn \cite{CL} gave algebraic cyclic code constructions of symbol-pair codes and asymptotic bounds on code rates. They also showed the existence of pair-error codes with rates  strictly higher than those of the codes in the Hamming metric with the same relative distance. Yaakobi et al.\ proposed efficient decoding algorithms for cyclic symbol-pair codes in \cite{YBS,YBH16}.

Chee et al.\ in \cite{CJKWY} established a Singleton-type bound on symbol-pair codes and constructed infinite families of symbol-pair codes that meet the Singleton-type bound, which are called maximum distance separable symbol-pair codes or MDS symbol-pair codes for short. The construction of MDS symbol-pair codes is interesting since the codes have the best pair-error correcting capability for fixed length and dimension. The authors in \cite{CJKWY} made use of interleaving and graph theoretic concepts as well as combinatorial configurations to construct MDS symbol-pair codes. Kai et al.\ \cite{KZL} constructed MDS symbol-pair codes from cyclic and constacyclic codes.

Classical MDS codes are MDS symbol-pair codes \cite{CJKWY} and other known families of MDS $(n,d)_{q}$ symbol-pair codes are shown in Table \ref{tab}.

\begin{table}[h]\caption{Known families of MDS symbol-pair codes\label{tab}}
\begin{center}
\begin{tabular}{|c|c|c|c|}
  \hline
  $d$ & $q$ & $n$ & Reference \\\hline
  $2$, $3$ & $q\geq2$ & $n\geq2$ &\cite{CJKWY}  \\\hline
  $4$& $q\geq 2$ & $n\geq 2$ & \cite{CJKWY} \\\hline
     & even prime power &$n\leq q+2$&\cite{CJKWY}\\\cline{2-4}
     & odd prime &$5\leq n\leq 2q+3$&\cite{CJKWY}\\\cline{2-4}
  $5$& prime power &$n|q^{2}-1$, $n>q+1$&\cite{KZL}\\\cline{2-4}
  & prime power &$n=q^{2}+q+1$&\cite{KZL}\\\cline{2-4}
  & prime power, $q\equiv 1\pmod 3$ &$n=\frac{q^{2}+q+1}{3}$&\cite{KZL}\\\hline
  \multirow{2}{*}{6}&prime power &$n=q^{2}+1$&\cite{KZL}\\\cline{2-4}
  &odd prime power& $n=\frac{q^{2}+1}{2}$ &\cite{KZL}\\\hline
  $7$ &odd prime&$n=8$&\cite{CJKWY}\\
  \hline
\end{tabular}
\end{center}
\end{table}

In this paper, we present new constructions of linear MDS symbol-pair codes over the finite field $\mathbb{F}_{q}$ and obtain the following three new families:
\begin{enumerate}
  \item there exists a linear MDS $(n,5)_q$ symbol-pair code if and only if $5\le n \le q^2+q+1$;
  \item there exists a linear MDS $(n,6)_{q}$ symbol-pair code for $q\geq3$ and $6\leq n\leq q^{2}+1$;
  \item there exists a linear MDS $(n,d+2)_{q}$ symbol-pair code for  general $n,d$ satisfying  $7\le d+2\leq n\le q+\lfloor 2\sqrt{q}\rfloor+\delta(q)-3$, where
      \[
\delta(q)=\left\{
  \begin{array}{ll}
    $0$, & \hbox{if $q=p^a,\,a\ge 3$, $a$ odd and $p\,|\,\lfloor 2\sqrt{q}\rfloor$;} \\
    $1$, & \hbox{otherwise.}
  \end{array}
\right.
\]
\end{enumerate}
Compared with the known MDS symbol-pair codes, the  MDS symbol-pair codes constructed in this paper provide a much larger range of parameters.

This paper is organized as follows. Basic notations and definitions are given in Section \ref{sec2}. In Section \ref{sec3}, we construct MDS symbol-pair codes with pair-distance $5$. And in Section \ref{sec4} we derive MDS symbol-pair codes with pair-distance $6$  from projective geometry. In Section \ref{sec5}, by using elliptic curves, we give the construction of MDS symbol-pair codes for any pair-distance satisfying certain conditions. Section \ref{sec6} concludes the paper.

\section{Preliminaries}\label{sec2}
Let $\Sigma$ be the alphabet consisting of $q$ elements.  Each element in $\Sigma$ is called a symbol. For a vector ${\bf u}=(u_{0},u_{1},\cdots,u_{n-1})$ in $\Sigma^{n}$, we define the symbol-pair read vector of ${\bf u}$ as
$$\pi({\bf u})=((u_{0},u_{1}),(u_{1},u_{2}),\cdots,(u_{n-1},u_{0})).$$

Throughout this paper, let $q$ be a prime power and $\mathbb{F}_{q}$ be the finite field containing $q$ elements. We will focus on vectors over $\mathbb{F}_{q}$, so $\Sigma=\mathbb{F}_{q}$.
It is obvious that each vector ${\bf u}$ in $\mathbb{F}_{q}^{n}$ has a unique symbol-pair read vector $\pi({\bf u})$ in $(\mathbb{F}_{q} \times \mathbb{F}_{q})^{n}$. For two vectors ${\bf u}$, ${\bf v}$ in $\mathbb{F}_{q}^{n}$, the pair-distance between ${\bf u}$ and ${\bf v}$ is defined as
$$d_{p}({\bf u},{\bf v}):=|\lbrace0\le i\le n-1:(u_{i},u_{i+1})\ne (v_{i},v_{i+1})\rbrace|,$$
where the subscripts are reduced modulo $n$. And for any vector ${\bf u}$ in $\mathbb{F}_{q}^{n}$, the pair-weight of ${\bf u}$ is defined as
$$w_{p}({\bf u})=|\lbrace0\le i\le n-1:(u_{i},u_{i+1})\ne (0,0)\rbrace|,$$ where the subscripts are reduced modulo $n$.

The following relationship between the pair-distance and the Hamming distance was shown in \cite{CB}.
\begin{proposition}\label{prop1}
  Let ${\bf u},{\bf v}\in \mathbb{F}_{q}^n$ be such that $0<d_{H}({\bf u},{\bf v})<n$, where $d_{H}$ denotes the Hamming distance, we have $$d_{H}({\bf u},{\bf v})+1\leq d_{p}({\bf u},{\bf v})\leq 2d_{H}({\bf u},{\bf v}).$$
\end{proposition}

Meanwhile, the following relationship between the pair-distance and the pair-weight holds.
\begin{proposition}\label{prop2}
  For all ${\bf u},{\bf v}\in \mathbb{F}_{q}^n$, $d_{p}({\bf u},{\bf v})=w_{p}({\bf u}-{\bf v}).$
\end{proposition}

A code $\mathcal{C}$ over $\mathbb{F}_{q}$ of length $n$ is a nonempty subset of $\mathbb{F}_{q}^{n}$ and the elements of $\mathcal{C}$ are called codewords. The minimum pair-distance of $\mathcal{C}$ is defined as
$$d_{p}(\mathcal{C})=\min\lbrace{d_{p}(\bf u},{\bf v})\mid{\bf u},{\bf v}\in \mathcal{C},{\bf u}\neq {\bf v}\rbrace,$$
and the size of $\mathcal{C}$ is the number of codewords it contains. In general, a code $\mathcal{C}$ over $\mathbb{F}_{q}$ of length $n$, size $M$ and minimum pair-distance $d$ is called an $(n,M,d)_{q}$ symbol-pair code. Besides, if $\mathcal{C}$ is a subspace of $\mathbb{F}_{q}^{n}$, then $\mathcal{C}$ is called a linear symbol-pair code. When $\mathcal{C}$ is a linear code, the minimum pair-distance of $\mathcal{C}$ is the smallest pair-weight of nonzero codewords of $\mathcal{C}$. And in this paper we consider linear symbol-pair codes over $\mathbb{F}_{q}$.

The minimum pair-distance $d$ is an important parameter in determining the error-correcting capability of $\mathcal{C}$. Thus it is significant to find symbol-pair codes of fixed length $n$ with pair-distance $d$ as large as possible. In \cite{CJKWY}, the authors  proved the following Singleton-type bound.
\begin{theorem}[Singleton bound]
Let $q\geq 2$ and $2\leq d\leq n$. If $\mathcal{C}$ is an $(n,M,d)_{q}$ symbol-pair code, then $M\leq q^{n-d+2}$.
\end{theorem}
A symbol-pair code achieving the Singleton bound is a maximum distance separable (MDS) symbol-pair code. An MDS $(n,M,d)_{q}$ symbol-pair code is simply called an MDS $(n,d)_{q}$ symbol-pair code.
%
%
%
In \cite{KZL}, the authors presented the following theorem.
\begin{theorem}\label{thmcite}
Let $\mathcal{C}$ be an  $[n,n-d_{H},d_{H}]$ linear code over $\mathbb{F}_{q}$. If the pair-distance $d\geq d_{H}+2$, then $\mathcal{C}$ is an MDS $(n,d_{H}+2)_{q}$ symbol-pair code.
\end{theorem}
Now we are ready to give a  sufficient condition for the existence of linear MDS symbol-pair codes in the following theorem.
\begin{theorem}\label{thmAMDS}
There exists a linear MDS $(n,d_{H}+2)_{q}$ symbol-pair code $\mathcal{C}$ if there exists a matrix with $d_{H}$ rows and $n\ge d_{H}+2\ge 4$ columns over $\mathbb{F}_{q}$, denoted by $H=[H_{0},H_{1},\cdots,H_{n-1}]$, where $H_{i}$  $(0\leq i\leq n-1)$ is the $i$-th column of $H$, satisfying:
\begin{itemize}
\item[1.] any $d_{H}-1$ columns of $H$ are linearly independent;
\item[2.] there exist $d_{H}$ linearly dependent columns;
\item[3.] any  $d_{H}$ cyclically consecutive columns are linearly independent, i.e., $H_{i},H_{i+1},\cdots,H_{i+d-1}$ are linearly independent for $0\leq i\leq n-1$, where the subscripts are reduced modulo $n$.
\end{itemize}
\end{theorem}
\begin{proof}
 Let $\mathcal{C}$ be the linear code with parity check matrix $H$. The first two conditions indicate that $\mathcal{C}$ is an $[n,n-d_{H},d_{H}]$ linear code with size $q^{n-d_{H}}$. Consider any codeword $c\in \mathcal{C}$ with $d_{H}$ nonzero coordinates. From Propositions \ref{prop1}, \ref{prop2} and the third condition, we can see that the $d_{H}$ nonzero coordinates are not in cyclically consecutive positions, and thus $w_{p}(c)\ge d_{H}+2$. For any other codeword $c'\in \mathcal{C}$, we must have the Hamming weight $w_{H}(c')\ge d_{H}+1$ and $w_{p}(c')\ge d_{H}+2$. Hence the pair-distance $d\ge d_{H}+2$ and $\mathcal{C}$ is an MDS $(n,d_{H}+2)_{q}$ symbol-pair code.
\end{proof}
\section{MDS symbol-pair codes with pair-distance 5}\label{sec3}
 We construct MDS $(n,5)_{q}$ symbol-pair codes in this section. According to Theorem \ref{thmAMDS}, what we need is to construct a matrix $H$ with $3$ rows and $n$ columns over $\mathbb{F}_{q}$ satisfying the following conditions:
\begin{itemize}
\item[1.] any two columns of $H$ are linearly independent;
\item[2.] there exist three linearly dependent columns;
\item[3.] any three cyclically consecutive columns are linearly independent.
\end{itemize}

\begin{lemma} \label{MDS5}
A linear MDS $(n,5)_{q}$ symbol-pair code, where $q$ is a prime power, exists only if the length $n$ ranges from $5$ to $q^2+q+1$.
\end{lemma}

\begin{proof}
From Proposition \ref{prop1}, we know that a symbol-pair code with the minimum pair-distance $d=5$ must have the minimum  Hamming distance $d_{H}\ge 3$. Thus the parity check matrix of the code must satisfy the first condition above and the conclusion follows.
\end{proof}

In this section we aim to show the existence of MDS $(n,5)_{q}$ symbol-pair codes for every $5\le n\le q^2+q+1$. We first describe how to construct a full matrix $H(q)$ of size $3\times (q^2+q+1)$ and then we mention how to adjust $H(q)$ to get a matrix $H(q;n)$ of size $3\times n$ for any $n$, $5\le n \le q^2+q+1$. Choose the column vectors of $H(q)$ from the following $q^2+q+1$ vectors: $\{(0,0,1)^{\textup{T}}$, $(0,1,c)^{\textup{T}}$ for $c\in \mathbb{F}_{q}$, $(1,a,b)^{\textup{T}}$ for $a,b\in\mathbb{F}_{q}\}$. In this way  the first two conditions above are guaranteed, and we only need to order these vectors in a proper way to meet the third condition.

First we deal with the case when $q$ is odd. Denote the elements in $\mathbb{F}_{q}$ in an arbitrary order $\{x_0,x_1,\dots,x_{q-1}\}$. As a preparatory step, we partition the $q^2$ vectors of the form $\{(1,a,b)^{\textup{T}},a,b\in\mathbb{F}_{q}\}$ into $q$ disjoint blocks $B_i=\{(1,a,a^2+x_i)^{\textup{T}},a\in\mathbb{F}_{q}\}$ for $0\le i <q$. We give an order of the vectors within $B_i$. Set the first vector to be $(1,x_i,x_i^2+x_i)^{\textup{T}}$, the next to be $(1,x_{i+1},x_{i+1}^2+x_i)^{\textup{T}}$, and then the next to be $(1,x_{i+2},x_{i+2}^2+x_i)^{\textup{T}}\dots$ until finally the vector $(1,x_{i+q-1},x_{i+q-1}^2+x_i)^{\textup{T}}$, where subscripts are reduced modulo $q$. That is,
\begin{equation*}
B_i=
\left[
\begin{array}{ccccc}
1 & 1 & 1 & \cdots & 1\\
x_i & x_{i+1} & x_{i+2} & \cdots & x_{i+q-1} \\
x_i^2+x_i & x_{i+1}^2+x_i & x_{i+2}^2+x_i & \cdots & x_{i+q-1}^2+x_i
\end{array}
\right].
\end{equation*}

Then we construct the matrix $H(q)$ as follows. List all the blocks $B_i$ defined above in the reverse order of their subscripts: $B_{q-1}$, $B_{q-2},\dots$, $B_1$, $B_0$. Between any pair of consecutive blocks $B_{i+1}$ and $B_{i}$, insert a vector $(0,1,2x_{i})^{\textup{T}}$. Note that the pair of $B_0$ and $B_{q-1}$ is also considered, and the vector $(0,1,2x_{q-1})^{\textup{T}}$ should be inserted between them, which is further restricted to be the first column of $H(q)$. Finally the vector $(0,0,1)^{\textup{T}}$ could be placed anywhere and we just set it as the last column. That is,
\begin{equation*}
H(q)=
\left[
\begin{array}{ccccccccccccccc}
0 &          & 0 &         & 0 &         & \dots &         &0  &    & \cdots &     & 0 &     &     0\\
1 & B_{q-1} & 1 & B_{q-2} & 1 & B_{q-3} & \dots & B_{i+1} &1  &B_i & \cdots & B_1 & 1 & B_0 &      0\\
2x_{q-1} &        & 2x_{q-2} &    & 2x_{q-3} &    & \dots &         &2x_i &    & \cdots &     & 2x_0 &     &  1
\end{array}
\right].
\end{equation*}
\begin{proposition} \label{independent}
Every three cyclically consecutive columns of $H(q)$ are linearly independent over $\mathbb{F}_q$.
\end{proposition}
\begin{proof}
For three consecutive columns within a block $B_i$, $0\leq i\leq q-1$, we have
\begin{equation*}\left|\begin{array}{ccc}1 & 1 & 1\\x_{a-1} & x_a & x_{a+1}\\ x_{a-1}^2+x_i & x_a^2+x_i & x_{a+1}^2+x_i \end{array}\right| = \left|\begin{array}{ccc}1 & 1 & 1\\x_{a-1} & x_a & x_{a+1}\\ x_{a-1}^2 & x_a^2 & x_{a+1}^2 \end{array}\right| =(x_{a-1}-x_a)(x_{a}-x_{a+1})(x_{a+1}-x_{a-1})\ne 0.\end{equation*}

For three consecutive columns with a vector $(0,1,2x_j)^{\textup{T}}$ in the middle, we have
\begin{equation*}\left|\begin{array}{ccc}1 & 0 & 1\\x_j & 1 & x_j\\ x_j^2+x_{j+1} & 2x_j & x_j^2+x_{j} \end{array}\right| = \left|\begin{array}{ccc}1 & 0 & 0\\x_j & 1 & 0\\ x_j^2+x_{j+1} & 2x_j & x_{j}-x_{j+1} \end{array}\right| = x_{j}-x_{j+1} \ne0.\end{equation*}

For three consecutive columns containing a vector $(0,1,2x_j)^{\textup{T}}$, which is not in the middle, we have either
\begin{equation}\label{check1}\left|\begin{array}{ccc}0 & 1 & 1\\1 & x_j & x_{j+1}\\ 2x_j & x_j^2+x_{j} & x_{j+1}^2+x_{j} \end{array}\right| = -(x_j-x_{j+1})^2 \ne0, \end{equation}
or
\begin{equation}\label{check2}\left|\begin{array}{ccc}1 & 1 & 0\\x_{j-1} & x_j & 1\\ x_{j-1}^2+x_{j+1} & x_j^2+x_{j+1} & 2x_j \end{array}\right| = (x_j-x_{j-1})^2 \ne0.\end{equation}

Finally, it is easy to see that every three consecutive columns in $H(q)$ containing the vector $(0,0,1)^{\textup{T}}$ are linearly independent over $\mathbb{F}_q$.
\end{proof}

We now focus on the case when $q$ is even and $q\neq 2,4$. The general outline is similar. Let $\omega$ be a primitive element in $\mathbb{F}_q$. Denote the elements in $\mathbb{F}_{q}$ in an arbitrary order $\{x_0,x_1,\dots,x_{q-1}\}$, with the only constraint that the first several elements are preset to be $x_0=0$, $x_1=1$, $x_2=\omega$, $x_3=\omega^2$, $x_4=\omega+1$, $x_5=\omega^2+\omega$. First define the blocks $B_i$ in the same way as above and list all the blocks $B_i$ in the reverse order of their subscripts: $B_{q-1},B_{q-2},\dots,B_1,B_0$. Now we need to find out which vector of the form $(0,1,y)^{\textup{T}}$ can be inserted between the blocks $B_{j+1}$ and $B_{j}$. Recall the proof of Proposition \ref{independent}. It can be checked that the choice of the value $y$ only affects equations (\ref{check1}) and (\ref{check2}). So for the validity of that proof we only require that
\begin{equation*}\left|\begin{array}{ccc}0 & 1 & 1\\1 & x_j & x_{j+1}\\ y & x_j^2+x_{j} & x_{j+1}^2+x_{j} \end{array}\right| = (x_{j+1}-x_j)(y-x_j-x_{j+1}) \ne0,\end{equation*}
and
\begin{equation*}\left|\begin{array}{ccc}1 & 1 & 0\\x_{j-1} & x_j & 1\\ x_{j-1}^2+x_{j+1} & x_j^2+x_{j+1} & y \end{array}\right| = (x_j-x_{j-1})(y-x_j-x_{j-1}) \ne0.\end{equation*}

That is, $y$ could be any value except for $x_j+x_{j-1}$ and $x_j+x_{j+1}$. An explicit insertion scheme seems hard to be expressed in an easy form, however, we can show that a proper insertion scheme surely exists. Construct a bipartite graph. The first part of the vertices corresponds to $\mathbb{F}_{q}$. The second part of the vertices is the set $\{L_j:0\le j < q\}$, where the symbol $L_j$ indicates the location between the blocks $B_{j+1}$ and $B_j$. $y\in\mathbb{F}_{q}$ is connected to $L_j$ if and only if the vector $(0,1,y)^{\textup{T}}$ could be inserted in the location $L_j$, i.e. $y\neq x_j+x_{j-1}$ and $y\neq x_j+x_{j+1}$. A perfect matching in this bipartite graph corresponds to a proper insertion scheme.

Following the analysis above, we can find that the degree of every vertex in the second part is exactly $q-2$. Recall that we have preset $x_0=0$, $x_1=1$, $x_2=\omega$, $x_3=\omega^2$, $x_4=\omega+1$, $x_5=\omega^2+\omega$. Thus we have:

$\bullet$ $L_1$ is connected to every $y\in\mathbb{F}_{q}$ except for $1$ and $\omega+1$;

$\bullet$ $L_2$ is connected to every $y\in\mathbb{F}_{q}$ except for $\omega+1$ and $\omega^2+\omega$;

$\bullet$ $L_3$ is connected to every $y\in\mathbb{F}_{q}$ except for $\omega^2+\omega$ and $\omega^2+\omega+1$; and

$\bullet$ $L_4$ is connected to every $y\in\mathbb{F}_{q}$ except for $\omega^2+\omega+1$ and $\omega^2+1$.

So, even only among these four vertices, we can deduce that every $y\in\mathbb{F}_{q}$ is connected to at least two of them. So we have

$\bullet$ the neighbourhood of every no more than $q-2$ vertices from the second part is of size at least $q-2$;

$\bullet$ the neighbourhood of every $q-1$ or $q$ vertices from the second part is of size $q$.

Therefore the famous Hall's theorem \cite{Hall} guarantees a perfect matching in this bipartite graph, which corresponds to a proper insertion scheme.

However, the case $q=4$ is listed as a separated case since the framework above using Hall's theorem  would fail. To follow a similar framework, the order within a block needs some slight modifications and then a proper insertion scheme comes along. We shall just list the desired $3\times21$ matrix $H(4)$ instead of tedious explanations.

{\scriptsize
\begin{equation*}
H(4)=
\left[
\begin{array}{ccccccccccccccccccccc}
0 & 1 & 1 & 1 & 1 & 0 & 1 & 1 & 1 & 1 & 0 & 1 & 1 & 1 & 1 & 0 & 1 & 1 & 1 & 1 & 0 \\
1 & 0 & 1 & \omega & \omega+1 & 1 & \omega+1 & \omega & 1 & 0 & 1 & 0 & \omega+1 & \omega & 1 & 1 & 1 & \omega & \omega+1 & 0 & 0 \\
0 & 0 & 1 & \omega+1 & \omega & \omega+1 & \omega+1 & \omega & 0 & 1 & \omega & \omega & 0 & 1 & \omega+1 & 1 & \omega & 0 & 1 & \omega+1 & 1
\end{array}
\right].
\end{equation*}
}

Up till now we have constructed the matrix $H(q)$ of size $3\times (q^2+q+1)$ for every prime power $q\ge3$. Next we discuss how to adjust $H(q)$ to get a $3\times n$ matrix $H(q;n)$ for every $n$, $5\le n\le q^2+q+1$. Denote $n=\alpha(q+1)+\beta$, where $0\le \beta \le q$. There are certainly lots of methods to get such a desired matrix and we offer one as follows.

$\bullet$ If $\beta\neq 2$, select the first $n-1$ columns of $H(q)$, then add the vector $(0,0,1)^{\textup{T}}$.

$\bullet$ If $\beta=2$, select the first $n-1$ columns of $H(q)$, then insert the vector $(0,0,1)^{\textup{T}}$ as the new third column.

The case $\beta=2$ is separated since if we still abide by the first rule then we will come across a triple of the form $\{(0,1,x)^{\textup{T}},(0,0,1)^{\textup{T}},(0,1,y)^{\textup{T}}\}$ which is certainly not independent.

The validity of the construction of the $3\times n$ matrix can be easily inferred from Proposition \ref{independent} plus some further simple checks on those triples containing the vector $(0,0,1)^{\textup{T}}$, and the two triples of the form $\{(0,1,a)^{\textup{T}},(0,1,b)^{\textup{T}},(1,c,d)^{\textup{T}}\}$ (in the $\beta=2$ case).

As illustrative examples, for $q=5$ we list the following matrices: the full matrix $H(5)$ of size $3\times 31$, the adjusted matrix $H(5;13)$ (corresponding to $\beta\neq 2$) and $H(5;14)$ (corresponding to $\beta=2$).

{\scriptsize
\begin{equation*}
H(5)=
\left[
\begin{array}{ccccccccccccccccccccccccccccccc}
0 & 1 & 1 & 1 & 1 & 1 & 0 & 1 & 1 & 1 & 1 & 1 & 0 & 1 & 1 & 1 & 1 & 1 & 0 & 1 & 1 & 1 & 1 & 1 & 0 & 1 & 1 & 1 & 1 & 1  & 0 \\
1 & 4 & 0 & 1 & 2 & 3 & 1 & 3 & 4 & 0 & 1 & 2 & 1 & 2 & 3 & 4 & 0 & 1 & 1 & 1 & 2 & 3 & 4 & 0 & 1 & 0 & 1 & 2 & 3 & 4  & 0 \\
3 & 0 & 4 & 0 & 3 & 3 & 1 & 2 & 4 & 3 & 4 & 2 & 4 & 1 & 1 & 3 & 2 & 3 & 2 & 2 & 0 & 0 & 2 & 1 & 0 & 0 & 1 & 4 & 4 & 1  & 1
\end{array}
\right],
\end{equation*}}
{\scriptsize
\begin{equation*}
H(5;13)=
\left[
\begin{array}{ccccccccccccc}
0 & 1 & 1 & 1 & 1 & 1 & 0 & 1 & 1 & 1 & 1 & 1 & 0 \\
1 & 4 & 0 & 1 & 2 & 3 & 1 & 3 & 4 & 0 & 1 & 2 & 0 \\
3 & 0 & 4 & 0 & 3 & 3 & 1 & 2 & 4 & 3 & 4 & 2 & 1
\end{array}
\right],
H(5;14)=
\left[
\begin{array}{cccccccccccccc}
0 & 1 & 0 & 1 & 1 & 1 & 1 & 0 & 1 & 1 & 1 & 1 & 1 & 0     \\
1 & 4 & 0 & 0 & 1 & 2 & 3 & 1 & 3 & 4 & 0 & 1 & 2 & 1   \\
3 & 0 & 1 & 4 & 0 & 3 & 3 & 1 & 2 & 4 & 3 & 4 & 2 & 4
\end{array}
\right].
\end{equation*}}

Finally, for the case $q=2$, we list the matrices $H(2)$, $H(2;5)$, $H(2;6)$ as follows.

\begin{equation*}
H(2)=
\left[
\begin{array}{ccccccc}
1 & 0 & 0 & 1 & 0 & 1 & 1 \\
0 & 1 & 0 & 1 & 1 & 1 & 0 \\
0 & 0 & 1 & 0 & 1 & 1 & 1
\end{array}
\right],
H(2;5)=
\left[
\begin{array}{ccccc}
1 & 0 & 0 & 1 & 1  \\
0 & 1 & 0 & 1 & 0  \\
0 & 0 & 1 & 1 & 1
\end{array}
\right],
H(2;6)=
\left[
\begin{array}{cccccc}
1 & 0 & 0 & 1 & 0 & 1  \\
0 & 1 & 0 & 1 & 1 & 0  \\
0 & 0 & 1 & 1 & 1 & 1
\end{array}
\right].
\end{equation*}

So far we have finished the construction of MDS $(n,5)_q$ symbol-pair codes for any prime power $q\ge2$ and $5\le n \le q^2+q+1$. The construction, together with Lemma \ref{MDS5}, leads to the following theorem.

\begin{theorem}
There exists a linear MDS $(n,5)_q$ symbol-pair code, where $q$ is a prime power, if and only if the length $n$ ranges from $5$ to $q^2+q+1$.
\end{theorem}
\section{MDS symbol-pair codes from projective geometry}\label{sec4}
 Let $V(r+1,q)$ be a vector space of rank $r+1$ over $\mathbb{F}_{q}$. The projective space $PG(r,q)$ is the geometry whose points, lines, planes, $\cdots$, hyperplanes are the subspaces of $V(r+1,q)$ of rank $1,2,3,\cdots,r$, respectively. The dimension of a subspace of $PG(r,q)$ is one less than the rank of a subspace of $V(r+1,q)$.

Label each point of $PG(r,q)$ as  $\langle(a_{0},a_{1},\cdots,a_{r})\rangle$, the subspace spanned by a nonzero vector $(a_{0},a_{1},\cdots,a_{r})$, where $a_{i}\in \mathbb{F}_{q}$ for $0\leq i\leq r$. Since these coordinates are defined only up to multiplication by a nonzero scalar $\lambda\in \mathbb{F}_{q}$ (here $\langle(\lambda a_{0},\lambda a_{1},\cdots,\lambda a_{r})\rangle=\langle(a_{0},a_{1},\cdots,a_{r})\rangle$), we refer to $a_{0},a_{1},\cdots,a_{r}$ as homogeneous coordinates.  Thus, there are a total of $(q^{r+1}-1)/(q-1)$ points in $PG(r,q)$. For an integer $r\geq 2$, if we choose $n\geq r+3$ points in $PG(r,q)$ and regard them as column vectors of a matrix $H$, then from Theorem \ref{thmAMDS} we have the following theorem.

\begin{theorem}\label{themcon}
There exists a linear MDS $(n,r+3)_{q}$ symbol-pair code if there exists a set $\mathcal{S}$ of $n\geq r+3\ge 5$ points of $PG(r,q)$ satisfying the following conditions:
\begin{itemize}
\item[1.] any $r$ points from $\mathcal{S}$ generate a hyperplane in $PG(r,q)$;
\item[2.] there exist $r+1$ points in $\mathcal{S}$ lying on a hyperplane;
\item[3.] if the $n$ points are ordered, say $\mathcal{P}_{0},\mathcal{P}_{1},\cdots,\mathcal{P}_{n-1}$, then any $r+1$ cyclically consecutive points do not lie on a hyperplane, i.e., $\mathcal{P}_{i},\mathcal{P}_{i+1},\cdots,\mathcal{P}_{i+r}$, where the subscripts are reduced modulo $n$, do not lie on a hyperplane for $0\leq i\leq n-1$.
\end{itemize}
\end{theorem}

Here we consider the case $r=3$.
\begin{definition}[]
A set $\mathcal{O}$ of points of $PG(3,q)$ is called an ovoid provided it satisfies the following conditions:
\begin{itemize}
\item[1.] each line meets $\mathcal{O}$ in at most two points;
\item[2.] through each point of $\mathcal{O}$ there are $q+1$ lines, each of which meets $\mathcal{O}$ in exactly one point, and all of them lie on a plane.
\end{itemize}
\end{definition}

The following two lemmas can be found in \cite{P09}.
\begin{lemma}
Each ovoid has $q^{2}+1$ points.
\end{lemma}

\begin{lemma}\label{lemPG1}
Each plane meets $\mathcal{O}$ either in one point or in $q+1$ points.
\end{lemma}

 We can easily derive the following lemma.
\begin{lemma}\label{thmPG2}
For an ovoid $\mathcal{O}$ in $PG(3,q)$, there exist $q+1$ planes, each of which contains $q+1$ points in $\mathcal{O}$. Moreover, these planes intersect in a common line in $\mathcal{O}$ and cover all points of $\mathcal{O}$.
\end{lemma}
\begin{proof}
Fix two arbitrary points $A,B \in \mathcal{O}$, and then choose a point $P$ from $\mathcal{O}\setminus \lbrace A,B\rbrace$. By Lemma \ref{lemPG1}, the plane formed by $A,B,P$, which we denote by $ABP$, must meet $\mathcal{O}$ in $q+1$ points. Next, choose a point $Q\in\mathcal{O}$ which is not on $ABP$. Then, again, we get
a plane $ABQ$ which also meets $\mathcal{O}$ in $q+1$ points. If we continue in this way, we can get $q+1$ planes, each of which contains $q+1$ points of $\mathcal{O}$. These planes intersect in a common line which meets $\mathcal{O}$ in the points $A,B$.
\end{proof}

We can now state our construction.
\begin{theorem}\label{thm1}
Let $q\geq5$ be an odd prime power. Then there exist linear MDS $(n,6)_{q}$ symbol-pair codes for all $n$, $6\leq n\leq q^{2}+1$.
\end{theorem}
\begin{proof} Let $\mathcal{O}$ be an ovoid in $PG(3,q)$ and $\pi_{0},\pi_{1},\cdots,\pi_{q}$ be the planes described in Lemma \ref{thmPG2}. Moreover, let the intersection of $\pi_{0},\pi_{1},\cdots,\pi_{q}$ meets $\mathcal{O}$ in the points $A$ and $B$. For convenience, denote the plane formed by points $P,Q,R$ by $PQR$ and denote the set of the points lying in a set, say $\Omega$, but not on the plane $PQR$ by $\Omega\setminus PQR$. For four ordered points $P,Q,R,S$, we say $S$ is a \emph{proper} point if $S$ does not lie on the plane $PQR$. In other words, we say $S$ is a \emph{proper} point if $S$ does not lie on the plane formed by the three points ordered right ahead of it.

We now consider the three conditions stated in Theorem \ref{themcon}. It is clear that, for the points of $\mathcal{O}$, the first condition is inherently satisfied and the second condition can be easily satisfied. Thus, the points of $\mathcal{O}$ simply need to be ordered such that any four cyclically consecutive points do not lie on a plane. To attain this goal, we discuss it in two parts. First we order $n$ ($6\leq n\leq q^{2}+1$) points of $\mathcal{O}$ as $\mathcal{P}_{0},\cdots,\mathcal{P}_{n-1}$ and make sure that any four consecutive points do not lie on a plane, i.e., $\mathcal{P}_{i},\mathcal{P}_{i+1},\mathcal{P}_{i+2},\mathcal{P}_{i+3}$ do not lie on a plane for $0\leq i\leq n-4$. On this basis, we then adjust the order to make sure that any four cyclically consecutive points do not lie on a plane, i.e., $\mathcal{P}_{i},\mathcal{P}_{i+1},\mathcal{P}_{i+2},\mathcal{P}_{i+3}$ do not lie on a plane for $0\leq i\leq n-1$.

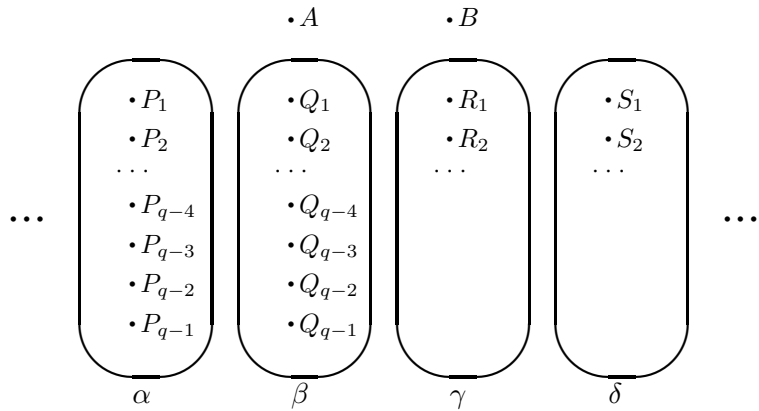
\begin{figure}[h]
\centering
\begin{picture}(200,160)
\thicklines
\multiput(20,115)(0,-15){2}{\circle*{2}}
\multiput(20,30)(0,15){4}{\circle*{2}}
\multiput(15,88)(5,0){3}{\circle*{1.2}}
\multiput(-25,70)(5,0){3}{\circle*{2}}
\multiput(245,70)(5,0){3}{\circle*{2}}
\put(20,0){$\alpha$}\put(80,0){$\beta$}\put(140,0){$\gamma$}\put(200,0){$\delta$}

\multiput(80,115)(0,-15){2}{\circle*{2}}
\multiput(80,30)(0,15){4}{\circle*{2}}
\multiput(75,88)(5,0){3}{\circle*{1.2}}

\put(80,145){\circle*{2}}
\put(140,145){\circle*{2}}
\put(83,143){\small $A$}
\put(143,143){\small $B$}

\multiput(140,115)(0,-15){2}{\circle*{2}}
\multiput(200,115)(0,-15){2}{\circle*{2}}
\multiput(135,88)(5,0){3}{\circle*{1.2}}
\multiput(195,88)(5,0){3}{\circle*{1.2}}

\multiput(25,70)(60,0){4}{\oval(50,120)}

\put(23,112){\small $P_{1}$}
\put(23,97){\small $P_{2}$}
\put(23,72){\small $P_{q-4}$}
\put(23,57){\small $P_{q-3}$}
\put(23,42){\small $P_{q-2}$}
\put(23,27){\small $P_{q-1}$}

\put(83,112){\small $Q_{1}$}
\put(83,97){\small $Q_{2}$}
\put(83,72){\small $Q_{q-4}$}
\put(83,57){\small $Q_{q-3}$}
\put(83,42){\small $Q_{q-2}$}
\put(83,27){\small $Q_{q-1}$}

\put(143,112){\small $R_{1}$}
\put(143,97){\small $R_{2}$}
\put(203,112){\small $S_{1}$}
\put(203,97){\small $S_{2}$}
\end{picture}
\caption {The sets $\pi_{i}\setminus \lbrace A,B\rbrace$ when $q$ is an odd prime power\label{fig1}.}
\end{figure}
First, let $\alpha,\beta,\gamma$ and $\delta$ denote the sets $\pi_{0}\setminus \lbrace A,B\rbrace,\pi_{1}\setminus \lbrace A,B\rbrace,\pi_{2}\setminus \lbrace A,B\rbrace,\pi_{3}\setminus \lbrace A,B\rbrace$ respectively, as illustrated in Figure \ref{fig1}. Let $A,B$ be the first and the second points. Choose an arbitrary point $P_{1}$ from $\alpha$ to be the third and an arbitrary point $Q_{1}$ from $\beta$ to be the fourth. It is obvious that $A,B,P_{1},Q_{1}$ do not lie on a plane. Next, choose $P_{2}\in\alpha\setminus BP_{1}Q_{1}$ to be the fifth and $Q_{2}\in\beta\setminus P_{1}P_{2}Q_{1}$ to be the sixth. Two planes intersect in a line and a line meets $\mathcal{O}$ in at most two points. Thus, we can continue in this way, i.e., take \emph{proper} points from $\alpha$ and $\beta$ in turn, until only one point remains in $\alpha$.

Now suppose this has been done so that the point $P_{q-1}$ remains, i.e., we have ordered the points as $A,B,P_{1},Q_{1},\cdots,P_{q-2},Q_{q-2}$. Then we have that $P_{q-4},Q_{q-4},P_{q-3},Q_{q-3}$ do not lie on a plane, nor do $Q_{q-4},P_{q-3},Q_{q-3},P_{q-2}$, and nor do $P_{q-3},Q_{q-3},P_{q-2},Q_{q-2}$. Next we order the two points $P_{q-1}$ and $Q_{q-1}$. We consider the following three cases:

{\bf Case 1:} $P_{q-1}\notin P_{q-2}Q_{q-3}Q_{q-2}$, $Q_{q-1}\notin P_{q-2}P_{q-1}Q_{q-2}$.

Note that this situation is ideal. Let the order be $P_{q-4},Q_{q-4},P_{q-3},Q_{q-3},P_{q-2},Q_{q-2},P_{q-1},Q_{q-1}$.

{\bf Case 2:} $P_{q-1}\notin P_{q-2}Q_{q-3}Q_{q-2}$, but $Q_{q-1}\in P_{q-2}P_{q-1}Q_{q-2}$.

Change the order to be $P_{q-4},Q_{q-4},P_{q-3},P_{q-2},Q_{q-3},Q_{q-2},P_{q-1},Q_{q-1}$.

{\bf Case 3:} $P_{q-1}\in P_{q-2}Q_{q-3}Q_{q-2}$.

Change the order to be $P_{q-4},Q_{q-4},P_{q-3},Q_{q-3},P_{q-2},Q_{q-2},Q_{q-1},P_{q-1}$.

Next, we find a \emph{proper} point $R_{1}\in\gamma$ to be the next point, as well as \emph{proper} points $S_{1}\in\delta$ and $R_{2}\in\gamma$. Then order the remaining points in $\gamma$ and $\delta$ just as what we have done for the points in $\alpha$ and $\beta$. Repeat the procedure  until we have covered $n$ ($6\leq n\leq q^{2}+1$)  points in $\mathcal{O}$. By now, we have got $n$ ordered points such that any four consecutive points do not lie on a plane.

Note that we have finished our first part. Denote the last four points by $W,X,Y$ and $Z$. To make sure that any four cyclically consecutive points do not lie on a plane, we still need to ensure that $X,Y,Z,A$ do not lie on a plane, nor do $Y,Z,A,B$ and nor do $Z,A,B,P_{1}$. We discuss in the following cases.

{\bf Case a:} $X,Y,Z$ and $A$ lie on a plane.

This happens only when $X\in\pi_{i}$, $Y\in\pi_{i+1}$ and $Z\in\pi_{i+2}$, for some $i$, $0\leq i\leq q-2$. For example, $P_{q-1},Q_{q-1}$ and $R_{1}$ in Figure \ref{fig1}. Otherwise, we always have exactly two of $X,Y,Z$ belonging to the same set $\pi_{j}\setminus\lbrace A,B\rbrace$, which ensures $X,Y,Z,A$ do not lie on a plane.

Note that $WXY$ intersects $\pi_{i+2}$ in at most two points and also $XYA$ intersects $\pi_{i+2}$ in at most two points, one of which is the point $A$. Thus, in this case, we find a new point $Z'$ in $\pi_{i+2}\setminus\lbrace A,B\rbrace$, not lying on planes $WXY$ and $XYA$ to be the new last point. We can always do this since there are totally $q+1\geq 6$ points on $\pi_{i+2}$.

{\bf Case b:} $Y,Z,A$ and $B$ lie on a plane.

This happens when the last two points lie in the same $\pi_{i}\setminus\lbrace A,B\rbrace$, which occurs in Cases $2$ and $3$ above. Note that $\alpha$ and $\beta$ can be any $\pi_{i}\setminus\lbrace A,B\rbrace$ and $\pi_{i+1}\setminus\lbrace A,B\rbrace$ respectively for $i=0,2,4,\cdots,q-1$ in the following discussion. In Case $2$, if the last three points are $Q_{q-4}$, $P_{q-3}$ and $P_{q-2}$, then we replace them by $Q_{q-4}$, $P_{q-3}$ and $Q_{q-3}$. If the last three points are $P_{q-2}$, $Q_{q-3}$ and $Q_{q-2}$,
then we replace them by $Q_{q-3}$, $P_{q-2}$ and $Q_{q-2}$. In Case $3$, if the last three points are $P_{q-2},Q_{q-2}$ and $Q_{q-1}$, then we replace them by $P_{q-2},P_{q-1}$ and $Q_{q-1}$.

{\bf Case c:} $Z,A,B$ and $P_{1}$ lie on a plane.

This happens when $Z$ lies in $\pi_{0}\setminus\lbrace A,B\rbrace$, i.e., $7\leq n\leq 2q-1$ and $n$ is odd.
In this case, after choosing the first three points $A,B,P_{1}$, we choose \emph{proper} points from $\pi_{2}\setminus\lbrace A,B\rbrace$ and $\pi_{3}\setminus\lbrace A,B\rbrace$ in turn.
\end{proof}

\begin{remark}
We use the condition that points $P_{q-4}$,$P_{q-3}$ and $P_{q-2}$ are on the same plane $\pi_{i}$, $0\le i\le q$, in Theorem {\rm \ref{thm1}}. Thus we exclude the case when $q=3$ since there are not enough points on each plane $\pi_{i}$. We give the MDS symbol-pair codes directly for $q=3$.  There exists a linear MDS $(n,6)_{3}$ symbol-pair code, $n\in\lbrace 6,7,8,9,10\rbrace$, whose parity check matrix is formed by the first $n$ columns of the matrix
$$\left[
  \begin{array}{cccccccccc}
    0&1&1&1&1&1&1&1&1&1\\
    1&0&1&2&1&2&2&1&2&1\\
    0&0&1&0&2&0&2&2&1&1\\
    0&0&1&1&2&2&1&0&2&0\\
  \end{array}
\right].$$
\end{remark}

\begin{theorem}\label{thm2}
Let $q\geq 8$ be an even prime power. Then there exist linear MDS $(n,6)_{q}$ symbol-pair codes for all $n$, $6\leq n\leq q^{2}+1$.
\end{theorem}
\begin{proof}
Let the notations be defined as in Theorem \ref{thm1}. Note that the case when $q$ is even is different from that when $q$ is odd due to there being an odd number of planes. For $6\leq n\leq q^{2}-q+2$, we can order $n$ points on $\pi_{0},\pi_{1},\cdots,\pi_{q-1}$ just as in Theorem \ref{thm1} since the number of planes is even. The key step of this proof is to put the remaining $q-1$ points in order. To attain this goal, we first order all the points of the first three planes, and then we can just proceed as the case when $q$ is odd.

Let $\alpha,\beta,\gamma,\delta,\zeta$ denote the sets $\pi_{0}\setminus \lbrace A,B\rbrace,\pi_{1}\setminus \lbrace A,B\rbrace,\pi_{2}\setminus \lbrace A,B\rbrace,\pi_{3}\setminus \lbrace A,B\rbrace,\pi_{4}\setminus \lbrace A,B\rbrace$ respectively, as illustrated in Figure \ref{fig2}. Again, let $A$ and $B$ be the first two points and choose arbitrary $P_{1}$ and $Q_{1}$ from $\alpha$ and $\beta$ respectively. Choose the next point $R_{1}\in\gamma\setminus BP_{1}Q_{1}$, and then $P_{2}\in\alpha \setminus P_{1}Q_{1}R_{1}$ and $Q_{2}\in\beta\setminus P_{2}Q_{1}R_{1}$, i.e., take \emph{proper} points from $\alpha,\beta$ and $\gamma$ in turn. We can continue in this way until only one point remains in $\alpha$.

Suppose this has been done so that $P_{q-1}$ remains, i.e., we have ordered the points as $A,B,P_{1},Q_{1}$,\\$R_{1},\cdots,P_{q-2},Q_{q-2},R_{q-2}$. Note that the intersection of two planes meets $\mathcal{O}$ in at most two points. We can always find a point $S_{1}$ in $\delta$ that does not lie on the planes $P_{q-2}Q_{q-2}R_{q-2}$ and $P_{q-1}Q_{q-2}R_{q-2}$ since the two planes intersect $\delta$ in at most four points and there are $q-1\geq 7$ points in $\delta$. Similarly, we can find $T_{1}\in \zeta$ not lying on planes $P_{q-1}R_{q-2}S_{1}$ and $P_{q-1}Q_{q-1}S_{1}$, $S_{2}\in \delta$ not lying on planes $P_{q-1}Q_{q-1}T_{1}$ and $Q_{q-1}R_{q-1}T_{1}$. Next find $T_{2}\in \zeta\setminus Q_{q-1}R_{q-1}S_{2}$, $S_{3}\in \delta \setminus R_{q-1}S_{2}T_{2}$ and $T_{3}\in \zeta\setminus R_{q-1}S_{3}T_{2}$. Let the order of points be $P_{q-2},Q_{q-2},R_{q-2},S_{1},P_{q-1},T_{1},Q_{q-1},S_{2},R_{q-1},T_{2},S_{3},T_{3}$. Note that we have ordered all the points in $\alpha,\beta$ and $\gamma$, and any four consecutive points do not lie on a plane. There are an even number of planes left. We can then simply proceed as in Theorem \ref{thm1}, and also the similar discussion follows that of Theorem \ref{thm1}.
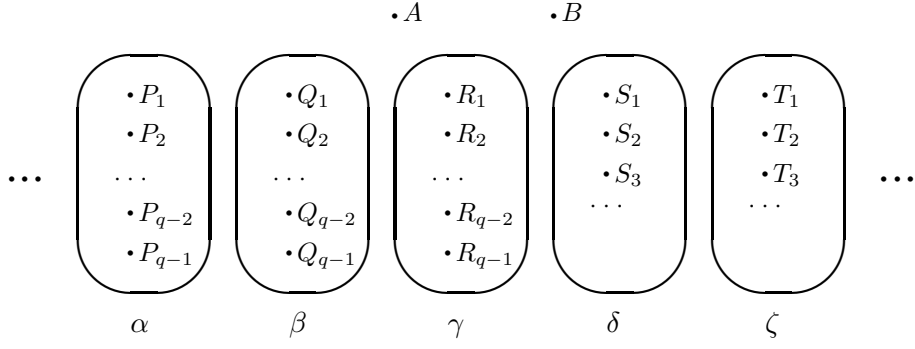
\begin{figure}[h]
\begin{center}
\begin{picture}(290,140)
\thicklines
\put(20,0){$\alpha$}\put(80,0){$\beta$}\put(140,0){$\gamma$}\put(200,0){$\delta$}\put(260,0){$\zeta$}
\multiput(25,60)(60,0){5}{\oval(50,90)}
\multiput(-25,58)(5,0){3}{\circle*{2}}
\multiput(305,58)(5,0){3}{\circle*{2}}

\multiput(20,90)(0,-15){2}{\circle*{2}}
\multiput(20,30)(0,15){2}{\circle*{2}}
\multiput(15,58)(5,0){3}{\circle*{1.2}}

\multiput(80,90)(0,-15){2}{\circle*{2}}
\multiput(80,30)(0,15){2}{\circle*{2}}
\multiput(75,58)(5,0){3}{\circle*{1.2}}

\multiput(140,90)(0,-15){2}{\circle*{2}}
\multiput(140,30)(0,15){2}{\circle*{2}}
\multiput(135,58)(5,0){3}{\circle*{1.2}}

\put(120,120){\circle*{2}}
\put(180,120){\circle*{2}}
\put(123,118){\small $A$}
\put(183,118){\small $B$}

\multiput(200,90)(0,-15){3}{\circle*{2}}
\multiput(260,90)(0,-15){3}{\circle*{2}}
\multiput(195,48)(5,0){3}{\circle*{1.2}}
\multiput(255,48)(5,0){3}{\circle*{1.2}}

\put(23,87){\small $P_{1}$}
\put(23,72){\small $P_{2}$}
\put(23,42){\small $P_{q-2}$}
\put(23,27){\small $P_{q-1}$}

\put(83,87){\small $Q_{1}$}
\put(83,72){\small $Q_{2}$}
\put(83,42){\small $Q_{q-2}$}
\put(83,27){\small $Q_{q-1}$}

\put(143,87){\small $R_{1}$}
\put(143,72){\small $R_{2}$}
\put(143,42){\small $R_{q-2}$}
\put(143,27){\small $R_{q-1}$}

\put(203,87){\small $S_{1}$}
\put(203,72){\small $S_{2}$}
\put(203,57){\small $S_{3}$}
\put(263,87){\small $T_{1}$}
\put(263,72){\small $T_{2}$}
\put(263,57){\small $T_{3}$}
\end{picture}
\end{center}\caption{The sets $\pi_{i}\setminus \lbrace A,B\rbrace$ when $q$ is an even prime power\label{fig2}.}
\end{figure}
\end{proof}

\begin{remark}
When $q=4$, since there are only five points on each plane $\pi_{i}$, $0\le i\le q$, we discuss it as a special case. Denote the primitive element of $\mathbb{F}_{4}$ as $w$. Then there is a linear MDS $(n,6)_{4}$ symbol-pair code, $n\in\lbrace6,8,9,10,11,12,13,14,15,16,17\rbrace$, and its parity check matrix is formed by the first $n$ columns of the matrix
$$\left[
    \begin{array}{ccccccccccccccccc}
      0&1&1&1&1  &1&1  &1&  1&1&1&  1&1&  1&  1&   1&  1\\
      1&0&1&w&1+w&1&w  &1+w&w&w&1+w&1&w&  1&  1+w& 1+w&1\\
      0&0&1&0&w  &0&1+w&0&  1&1&w&  w&w+1&w&  1+w& 1+w&1\\
      0&0&0&1&0  &w&0  &1+w&1&w&1&  w&w&  1+w&1+w& 1  &1+w\\
    \end{array}
  \right].
$$
There exists a linear MDS $(7,6)$ symbol-pair code with parity check matrix
$$\left[
    \begin{array}{ccccccc}
      0&1&1&1&1  &1&1\\
      1&0&1&w&1+w&1&w\\
      0&0&1&0&w  &0&1\\
      0&0&0&1&0  &w&w\\
    \end{array}
  \right].
$$
\end{remark}

Summing up the above, we can conclude the following theorem.
\begin{theorem}
For any prime power $q$, $q\geq 3$, and any integer $n$, $6 \leq n\leq q^{2}+1$, there exists a linear MDS $(n,6)_{q}$ symbol-pair
code.
\end{theorem}

\begin{remark}
Compare the cases $r=2$ and $r=3$ in Theorem {\rm \ref{themcon}}, if we consider the set of all the points instead of the ovoid, all the lines through a fixed point instead of the planes described in Lemma {\rm \ref{thmPG2}}, then we can also get linear MDS $(n,5)_{q}$ symbol-pair codes for $5\leq n\leq q^{2}+q+1$ with $q$ being a prime power in a similar way. Thus, this method deserves further investigation for larger $r$, which may derive MDS symbol-pair codes with larger pair-distance. 
\end{remark}

\section{MDS symbol-pair codes from elliptic curves}\label{sec5}
The previous two sections construct MDS symbol-pair codes with pair-distance $5$ and $6$. In this section, we give a construction of MDS symbol-pair codes with general pair-distance ($\ge 7$) from elliptic curve codes. We first briefly review some facts about elliptic curve codes.

Let $E/\f{q}$ be an elliptic curve over  $\f{q}$
with function field $\f{q}(E)$. Let $E(\f{q})$ be the set of all $\f{q}$-rational points on $E$.
Suppose $D=\{P_{1},P_{2},\cdots,P_{n}\}$ is a proper subset of rational points $E(\f{q})$, and $G$ is a divisor of degree $k$
($2g-2<k<n$) with $\mathrm{Supp}(G)\cap D=\emptyset$.
Without any confusion, we also write $D=P_{1}+P_{2}+\cdots+P_{n}$.
Denote by $\mathscr{L}(G)$ the $\f{q}$-vector space of all
rational functions $f\in \f{q}(E)$ with the principal divisor
$\mathrm{div}(f)\geqslant -G$, together with the zero function
(\cite{Stichtenoth}).

The functional AG code $C_{\mathscr{L}}(D, G)$ is defined to be the image of the following evaluation map:
 \[
       ev: \mathscr{L}(G)\rightarrow \f{q}^{n};\, f\mapsto
       (f(P_{1}),f(P_{2}),\cdots,f(P_{n})).\enspace
\]

It is well-known that $C_{\mathscr{L}}(D, G)$ is a linear code with parameters $[n,k,d_{H}]$, where the minimum Hamming distance $d_{H}$ has two choices:
\[
  d_{H}=n-k,\,\mbox{ or }\, d_{H}=n-k+1.\
\]
A linear $[n,k,d_{H}]$ code is called an MDS code if $d_{H}=n-k+1$ and is called an almost MDS code if $d_{H}=n-k$.

Suppose $O$ is one of the $\f{q}$-rational
points on $E$. The set of rational points $E(\f{q})$ forms an abelian group
with zero element $O$ (for the definition of the sum of any two
points, we refer to \cite{Silverman09}), and it is isomorphic to the
Picard group
 $\mathrm{div}^o(E)/\mathrm{Prin}(\f{q}(E))$, where $\mathrm{Prin}(\f{q}(E))$
 is the subgroup consisting of all principal divisors. 

 Denote by $\oplus$ and $\ominus$ the additive and minus operator in the group $E(\f{q})$, respectively.

\begin{proposition}[\cite{Cheng08,ZFW14}]\label{ECC}
Let $E$ be an elliptic curve over $\f{q}$ with an $\f{q}$-rational point $O$,
$D=\{P_{1},P_{2},\cdots,P_{n}\}$ a subset of $E(\f{q})$ such that
 $O\notin D$ and let $G=kO$ ($0<k<n$). Endow $E(\f{q})$ a group structure with the zero element $O$.
Denote by
\[
   N(k,O,D)=|\{S\subset D\,:\,|S|=k,\,\oplus_{P\in S}P=O\}|.
\]
Then the AG code $C_{\mathscr{L}}(D, G)$ has the minimum Hamming distance $d_H=n-k+1$ if and only if
\[
  N(k,O,D)= 0.
\]
And the minimum Hamming distance $d_H=n-k$ if and only if
\[
  N(k,O,D)>0.
\]

\end{proposition}
\begin{proof}
We have already seen that the minimum distance of $C_{\mathscr{L}}(D, G)$ has two choices: $n-k$, $n-k+1$. So $C_{\mathscr{L}}(D, G)$ is not MDS, i.e., $d=n-k$ if and only if there is a function $f\in\mathscr{L}(G)$ such that the evaluation $ev(f)$ has weight $n-k$. This is equivalent to that $f$ has $k$ zeros in $D$, say $P_{i_1}, \cdots, P_{i_k}$. That is
\[
    \mathrm{div}(f)\geq -(k-1)O-P+(P_{i_1}+\cdots+P_{i_k}),
\]
which is equivalent to
\[
    \mathrm{div}(f)=-(k-1)O-P+(P_{i_1}+\cdots+P_{i_k}).
\]
The existence of such an $f$ is equivalent to saying
\[
  P_{i_1}\oplus\cdots\oplus P_{i_k}=P.
\]
Namely, $N(k,P,D)> 0.$
It follows that the AG code $C_{\mathscr{L}}(D, G)$ has the minimum Hamming distance $n-k+1$ if and only if
$ N(k,P,D)=0.\ $
\end{proof}

 We restrict to the case $n> q+1$, since for every length $n\le q+1$, MDS symbol-pair codes of length $n$ can be constructed from Reed-Solomon codes. In this case, the minimum Hamming distance $d_H$ of elliptic curve codes is related to the main conjecture of MDS codes which was affirmed for elliptic curve codes~\cite{LWZ15,Mun92}.

\begin{proposition}[\cite{LWZ15,Mun92}]\label{ECSSP}
Let $C_{\mathscr{L}}(D, G)$ be the elliptic curve code constructed in Proposition {\rm \ref{ECC}} with length $n> q+1$. Then the subset sum problem always has solutions, i.e.,
\[
  N(k,O,D)>0.
\]
And hence, elliptic curve codes with length $n> q+1$ have deterministic minimum Hamming distance $d_H=n-k$.
\end{proposition}
That is, elliptic curve codes with length $n> q+1$ are almost MDS codes.
To obtain long codes from elliptic curves, we need the following two well-known results of elliptic curves over finite fields.
\begin{lemma}[Hasse-Weil Bound~\cite{Silverman09}]\label{HasseWeil}
Let $E$ be an elliptic curve over $\f{q}$. Then the number of $\f{q}$-rational points on $E$ is bounded by
\[
    |E(\f{q})|\le q+\lfloor 2\sqrt{q}\rfloor+1.
\]
\end{lemma}

\begin{lemma}[Hasse-Deuring~\cite{Deuring41}]\label{HasseDeuring}
The maximal number $N(\f{q})$ of $\f{q}$-rational points on $E$, where $E$ runs over all elliptic curves over $\f{q}$, is
\[
N(\f{q})=\left\{
 \begin{array}{ll}
 q+\lfloor 2\sqrt{q}\rfloor, &\hbox{if $q=p^a,\,a\ge 3$, $a$ odd and $p|\lfloor 2\sqrt{q}\rfloor$};\\
 q+\lfloor 2\sqrt{q}\rfloor+1, &\hbox{otherwise.}
 \end{array}
\right.
\]
\end{lemma}
Denote by
\[
\delta(q)=\left\{
  \begin{array}{ll}
    0, & \hbox{if $q=p^a,\,a\ge 3$, $a$ odd and $p\,|\,\lfloor 2\sqrt{q}\rfloor$;} \\
    1, & \hbox{otherwise.}
  \end{array}
\right.
\]

To construct an MDS symbol-pair code from classical error-correcting codes with large minimum Hamming distance, the key step is to find a way of ordering the coordinates. For general codes, this step seems very difficult. In the rest of this paper, we deal with the case of elliptic curve codes.

\begin{theorem}
Let $N(\f{q})=q+\lfloor 2\sqrt{q}\rfloor+\delta(q)$. Then for any $7\le d+2\leq n\le N(\f{q})-3$, there exist linear
MDS symbol-pair  codes  over $\f{q}$  with parameters $(n,d+2)_{q}$.

\end{theorem}
\begin{proof}
The existence of MDS symbol-pair codes with parameters $d+2=n$ follows from~\cite{CJKWY}. Below we  only consider the case $7\le d+2< n\le N(\f{q})-3$. By Lemma~\ref{HasseDeuring}, take $E$ to be a maximal elliptic curve over $\f{q}$ with an $\f{q}$-rational point $O$, i.e.,
\[
   |E(\f{q})|=N(\f{q}).
\]
Take divisor $G=kO$ in the construction of elliptic curve codes.

Case (I):  $N=N(\f{q})$ is odd, then there is no element of order $2$ in $E(\f{q})$. Suppose
$$E(\f{q})=\{P_{1},P_{2},\cdots,P_{N-2}, P_{N-1}, O\},$$
  where $P_1\oplus P_2=P_3\oplus P_4=\cdots=P_{N-2}\oplus P_{N-1}=O$.
  \begin{enumerate}
    \item For odd $d$ and even  $n:\,7\le d+2<n\le N-1$, in this case $k=N-1-d$ is odd. Take
    \[
       D=\{P_{1},P_{2},\cdots,P_{N-2}, P_{N-1}\}.
    \]
    Then by Proposition~\ref{ECSSP}, there are no $k$ cyclically consecutive points whose sum is $O$. And hence, the elliptic curve code $C_{\mathscr{L}}(D, G)$ is an MDS symbol-pair code with parameters $(N-1,d+2)_{q}$. By deleting pairs $(P_1,P_2),(P_3,P_4)$, etc., we can obtain MDS symbol-pair  codes with parameters $(n,d+2)_{q}$, where $n$ runs over all even integers $7\le d+2<n\le N-1$.

     \item For even $d$ and odd  $n:\,7\le d+2<n\le N-2$, in this case $k=N-2-d$ is odd. Take
    \[
       D=\{P_{1},P_{2},\cdots,P_{N-2}\}.
    \]
    Then by Proposition~\ref{ECSSP}, there are no $k$ cyclically consecutive points whose sum is $O$. And hence,  the elliptic curve code $C_{\mathscr{L}}(D, G)$ is an MDS symbol-pair code with parameters $(N-2,d+2)_{q}$. By deleting pairs $(P_1,P_2),(P_3,P_4)$, etc., we can obtain MDS symbol-pair  codes with parameters $(n,d+2)_{q}$ where $n$ runs over all odd integers $d+2<n\le N-2$.
    \item For even $d$ and even  $n:\,7\le d+2<n\le N-3$, in this case $k=N-3-d$ is even. Write $N-3=(k+1)s+r$ for some integers $s\ge 1$ and $0\le r\le k$.
    Take the pre-evaluation set
    \[
      D_0=\{P_{1},P_{2},\cdots,P_{N-5},P_{N-4},P_{N-2}\}
    \]
    and arrange it by the following algorithm:

    \textbf{Step 1.} Arrange $D_0$ as following
\begin{equation*}
    \begin{array}{rl}
      D_1=&\{P_{1},\cdots,P_{k-1},P_{N-5},P_k,\cdots,P_{(s-1)k-1},P_{N-3-s},P_{(s-1)k},\\
       &\cdots,P_{sk-1},P_{N-4},P_{sk}, P_{sk+1},\cdots,P_{sk+r-1},P_{N-2}\}.\\
    \end{array}
\end{equation*}

    After this step, there are no $k$ consecutive  points whose sum is $O$ in the sequence
    \[
      P_{1},\cdots,P_{k-1},P_{N-5},P_k,\cdots,P_{(s-1)k-1},P_{N-3-s},P_{(s-1)k},\cdots,P_{sk-1},P_{N-4},P_{sk}, P_{sk+1},\cdots,P_{sk+r-2}.
    \]
    But there may be some $k$ cyclically consecutive points whose sum is $O$ in the tail sequence
    \[
     P_{(s-1)k+r+1},\cdots,P_{sk-1},P_{N-4},P_{sk}, P_{sk+1},\cdots,P_{sk+r-1},P_{N-2}, P_{1},\cdots,P_{k-r-1}.
    \]
    For instance, $k=6$, $N=19$, by Step~1, we get
    \[
     D_1={P_1,\cdots,P_5,P_{14},P_6,\cdots,P_{11},P_{15},P_{12},P_{13},P_{17}}.
    \]
    There are no $6$ consecutive  points whose sum is $O$ in the sequence
    \[
    P_1,\cdots,P_5,P_{14},P_6,\cdots,P_{11},P_{15},P_{12},P_{13}.
    \]
    But there may be some $6$ cyclically consecutive  points whose sum is $O$ in the tail sequence
    \[
      P_{10},P_{11},P_{15},P_{12},P_{13},P_{17},P_1,P_2.
    \]

    \textbf{Step 2.} In the case that $r$ is even. It is easy to see that at most one of the following two equalities holds:
     \[
       P_{(s-1)k+r+2}\oplus\cdots\oplus P_{N-4}\oplus\cdots\oplus P_{N-2}=P_{(s-1)k+r+2}\oplus P_{N-4}\oplus P_{sk+r-1}\oplus P_{N-2}=O,
    \]
    and
     \[
       P_{(s-1)k+r+3}\oplus\cdots\oplus P_{N-4}\oplus\cdots\oplus P_{N-2}\oplus P_1= P_{N-4}\oplus P_{sk+r-1}\oplus P_{N-2}\oplus P_1=O.
    \]
    If the first one holds, then SWITCH $P_{(s-1)k+r+1}$ and $P_{(s-1)k+r+2}$; if the second one holds, then SWITCH $P_1$ and $P_2$; if neither of the two
    holds, then do nothing.

       For any $i=1,\cdots,\frac{k-r-2}{2}$, similarly at most one of the following two equalities holds:
        \[
       P_{(s-1)k+r+2i+2}\oplus\cdots\oplus P_{N-4}\oplus\cdots\oplus P_{N-2}\oplus P_{1}\oplus \cdots\oplus P_{2i}=P_{(s-1)k+r+2i+2}\oplus P_{N-4}\oplus P_{sk+r-1}\oplus P_{N-2}=O,
    \]
    and
    \[
       P_{(s-1)k+r+2i+1}\oplus\cdots\oplus P_{N-4}\oplus\cdots\oplus P_{N-2}\oplus P_{1}\oplus \cdots\oplus P_{2i-1}= P_{N-4}\oplus P_{sk+r-1}\oplus P_{N-2}\oplus P_{2i+1}=O.
    \]
    If the first one holds, then SWITCH $P_{(s-1)k+r+2i+1}$ and $P_{(s-1)k+r+2i+2}$; if the second one holds, then SWITCH $P_{2i+1}$ and $P_{2i+2}$; if neither of the two
    holds, then do nothing.

    In the case that $r$ is odd, the algorithm is the same as the even case, check the sum of $k$  cyclically consecutive points and do the corresponding SWITCH operation.

    Continue the above example, if
    \[
      P_{10}\oplus P_{11}\oplus P_{15}\oplus P_{12}\oplus P_{13}\oplus P_{17}=P_{10}\oplus P_{15}\oplus P_{13}\oplus P_{17}=O,
    \]
    then SWITCH $P_{9}$ and $P_{10}$; and in this case, it is immediate that
    \[
    P_{11}\oplus P_{15}\oplus P_{12}\oplus P_{13}\oplus P_{17}\oplus P_1=P_{15}\oplus P_{13}\oplus P_{17}\oplus P_1\neq O,
    \]
    so we do not need to reorder $P_1$ and $P_2$, and so on.

    Using the above algorithm to rearrange the evaluation set to get a newly arranged evaluation set $D$, by Proposition~\ref{ECSSP}, there are no $k$ cyclically consecutive  points whose sum is $O$. And hence, the elliptic curve code $C_{\mathscr{L}}(D, G)$ is an MDS symbol-pair code with parameters $(N-3,d+2)_{q}$. So, similarly as above, by deleting pairs from the pre-evaluation set, we can obtain MDS symbol-pair codes with parameters $(n,d+2)_{q}$ where $n$ runs over all even integers $d+2<n\le N-3$.

        \item For odd $d$ and odd  $n:\,7\le d+2<n\le N-2$, in this case $k=N-2-d$ is even. Write $N-2=(k+1)s+r$ for some integers $s\ge 1$ and $0\le r\le k$.
    Take the pre-evaluation set
    \[
      D_0=\{P_{1},P_{2},\cdots,P_{N-3},P_{N-2}\}
    \]
    and arrange it as following
        \begin{equation*}
        \begin{array}{rl}
          D=&\{P_{1},\cdots,P_{k-1},P_{N-3},P_k,\cdots,P_{(s-1)k-1}, \\
           & P_{N-1-s},P_{(s-1)k},\cdots,P_{sk-1},P_{N-2},P_{sk}, P_{sk+1},\cdots,P_{sk+r}\}.
        \end{array}
    \end{equation*}
    If $r$ is even, then it is easy to see that by Proposition~\ref{ECSSP} there are no $k$ cyclically consecutive points whose sum is $O$.

    If $r$ is odd, then similarly as the case when $d$ and $n$ are even, there may be some $k$ cyclically consecutive points whose sum is $O$ in the tail sequence. In this case, we just need process the same algorithm in the case $3$ to obtain a rearranged evaluation set $D$ such that there are no $k$ cyclically consecutive points whose sum is $O$.

    And hence,  the elliptic curve code $C_{\mathscr{L}}(D, G)$ is an MDS symbol-pair code with parameters $(N-2,d+2)_{q}$. So, similarly as above, by deleting pairs from the pre-evaluation set, we can obtain MDS symbol-pair  codes with parameters $(n,d+2)_{q}$ where $n$ runs over all odd integers $7\le d+2<n\le N-2$.
    \end{enumerate}
In conclusion, in the case that $N=N(\f{q})$ is odd, for any $7\le d+2\le n\le N(\f{q})-3$, no matter whether $d$ is odd or even, there exists an MDS
symbol-pair code with parameters $(n,d+2)_{q}$.

Case (II): $N=N(\f{q})$ is even. The proof is the same. Note that there are one or three non-zero elements of order $2$ in the group $E(\f{q})$.
Using these elements in the setting of the pre-evaluation set, the remainder of the argument is analogous. We omit the details here.

So, by the discussion above, we complete the proof of the theorem.
\end{proof}

 \begin{remark}
 From the proof, we see that in some cases, the length of the MDS symbol-pair code constructed from elliptic curve can attain $N(\f{q})-2$ or $N(\f{q})-1$. We omit the detail of the statements in the theorem to get a clear description of our result. Also, note that there are other works devoted to constructing almost MDS  codes using curves {\rm \cite{BC}} besides elliptic curves. To construct MDS symbol-pair codes using these almost MDS codes, how to arrange the evaluation set becomes the difficult step.
\end{remark}

\section{Conclusion}\label{sec6}
In this paper, we first  give a sufficient condition for the existence of linear MDS symbol-pair codes over $\mathbb{F}_{q}$.
On this basis, we show that a linear MDS $(n,5)_q$ symbol-pair code over $\mathbb{F}_{q}$ exists if and only if the length $n$ ranges from $5$ to $q^2+q+1$. Later, we introduce a special configuration in projective geometry called ovoid, which allows us to derive $q$-ary linear MDS symbol-pair codes with $d=6$ and length ranging from $6$ to $q^{2}+1$. This is an interesting method and deserves further investigation since it works well for both $d=5$ and $d=6$, and it may work for larger pair-distance. With the help of elliptic curves, we show that we can construct linear MDS $(n,d+2)_{q}$ symbol-pair codes for any $n,d$ satisfying $7\le d+2\le n\le q+\lfloor 2\sqrt{q}\rfloor+\delta(q)-3$. Compared with the results listed in Table \ref{tab}, our results provide a much larger range of parameters.

\end{document}